\documentclass[letterpaper]{article}
\usepackage[papersize={8.5in,11in},top=1.3in,left=1.3in,right=1.3in,bottom=1.3in]{geometry}
\usepackage[final]{graphicx}
\usepackage{amssymb,amsthm}
\usepackage{amsmath}
\usepackage[boxed]{algorithm2e}
\usepackage{appendix}
\usepackage{bm}
\usepackage{url}
\usepackage{hyperref}
\usepackage{multirow}

\newenvironment{definition}[1][Definition]{\begin{trivlist}
\item[\hskip \labelsep {\bfseries #1}]}{\end{trivlist}}

\newtheorem{theorem}{Theorem}[section]

\newtheorem{lemma}[theorem]{Lemma}
\newtheorem{claim}[theorem]{Claim}

\newtheorem{corollary}[theorem]{Corollary}

\newcommand\floor[1]{\lfloor#1\rfloor}
\newcommand\ceil[1]{\lceil#1\rceil}

\title{Blind, Greedy, and Random: Ordinal Approximation Algorithms for Matching and Clustering\thanks{A subset of the results in this paper appeared in the Proceedings of AAAI'$16$}}
\author{Elliot Anshelevich \and Shreyas Sekar}
\begin{document}
\maketitle

\begin{abstract}
We study Matching and other related problems in a partial information setting where the agents' utilities for being matched to other agents are hidden and the mechanism only has access to ordinal preference information. Our model is motivated by the fact that in many settings, agents cannot express the numerical values of their utility for different outcomes, but are still able to rank the outcomes in their order of preference. Specifically, we study problems where the ground truth exists in the form of a weighted graph, and look to design algorithms that approximate the true optimum matching using only the preference orderings for each agent (induced by the hidden weights) as input. If no restrictions are placed on the weights, then one cannot hope to do better than the simple greedy algorithm, which yields a half optimal matching. Perhaps surprisingly, we show that by imposing a little structure on the weights, we can improve upon the trivial algorithm significantly: we design a $1.6$-approximation algorithm for instances where the hidden weights obey the metric inequality. Using our algorithms for matching as a black-box, we also design new approximation algorithms for other closely related problems: these include a a $3.2$-approximation for the problem of clustering agents into equal sized partitions, a $4$-approximation algorithm for Densest $k$-subgraph, and a $2.14$-approximation algorithm for Max TSP. These results are the first non-trivial ordinal approximation algorithms for such problems, and indicate that we can design robust algorithms even when we are agnostic to the precise agent utilities.

\end{abstract}

\section{Introduction}
Consider the \emph{Maximum Weighted Matching} (MWM) problem, where the input is an undirected complete graph $G=(\mathcal{N},E)$ and the weight of an edge $w(i,j)$ represents the utility of matching agent $i$ with agent $j$. The objective is to form a matching (collection of disjoint edges) that maximizes the total utility of the agents. The problem of matching agents and/or items is at the heart of a variety of diverse applications and it is no surprise that this problem and its variants have received extensive consideration in the algorithmic literature~\cite{lovasz2009matching}. Perhaps, more importantly, maximum weighted matching is one of the few non-trivial combinatorial optimization problems that can be solved optimally in poly-time~\cite{edmonds1965paths}. In comparison, we study the MWM problem in a partial information setting where the lack of precise knowledge regarding agents' utilities acts as a barrier against computing optimal matchings, efficiently or otherwise.

More generally, in this work, we also look at other graph optimization problems such as clustering in a similar partial information setting, where optimal computation is preemptively stymied by the NP-Hardness of the problem (even in the full information case). This includes the problem of clustering agents to maximize the total weight of edges inside each cluster (Max $k$-sum), Densest $k$-subgraph, and the max traveling salesman problem. Furthermore, for the majority of this work, we assume that the edge weights obey the triangle inequality, since in many important applications it is natural to expect that the weights have some geometric structure. Such structure occurs, for instance, when the agents are points in a metric space and the weight of an edge is the distance between the two endpoints.

\subsubsection*{Partial Information - Ordinal Preferences}
A crucial question in algorithm and mechanism design is: ``How much information about the agent utilities does the algorithm designer possess?". The starting point for the rest of our paper is the observation that in many natural settings, it is unreasonable to expect the mechanism to know the exact weights of the edges in $G$~\cite{boutilierCHLPS15,procacciaR06}. For example, when pairing up students for a class project, it may be difficult to precisely quantify the synergy level for every pair of students; ordinal questions such as `who is better suited to partner with student $x$: $y$ or $z$?' may be easier to answer. Such a situation would also arise when the graph represents a social network of agents, as the agents themselves may not be able to express `exactly how much each friendship is worth', but would likely be able to form an ordering of their friends from best to worst. This phenomenon has also been observed in social choice settings, in which it is much easier to obtain ordinal preferences instead of true agent utilities~\cite{anshelevichBP15,procacciaR06}.

Motivated by this, we consider a model where for every agent $i \in \mathcal{N}$, we {\em only have access to a preference ordering} among the agents in $\mathcal{N} - \left\lbrace i \right\rbrace$ so that if $w(i,j) > w(i,k)$, then $i$ prefers $j$ to $k$. The common approach in Learning Theory while dealing with such ordinal settings is to estimate the `true ground state' based on some probabilistic assumptions on the underlying utilities~\cite{ohS14,soufianiPX12}. In this paper we take a different approach, and instead focus on the more demanding objective of designing \emph{robust algorithms}, i.e., algorithms that provide good performance guarantees no matter what the underlying weights are.

Despite the large body of literature on computing matchings in settings with preference orderings, there has been much less work on quantifying the quality of these matchings. As is common in much of \emph{social choice theory}, most papers (implicitly) assume that the underlying utilities cannot be measured or do not even exist, and hence there is no clear way to define the quality of a matching~\cite{abrahamIKM07,bhalgatCK11a,gusfield1989stable}. In such papers, the focus therefore is on computing matchings that satisfy normative properties such as stability or optimize a measure of efficiency that depends only on the preference orders, e.g., average rank. On the other hand, the literature on approximation algorithms usually follows the \emph{utilitarian} approach~\cite{harsanyi1976cardinal} of assigning a numerical quality to every solution; the presence of input weights is taken for granted. Our work combines the best of both worlds: we do {\em not} assume the availability of numerical information (only its latent existence), and yet our approximation algorithms must compete with algorithms that know the true input weights.

\subsubsection*{Model and Problem Statements}
For all of the problems studied in this paper, we are given as input a set $\mathcal{N}$ of points or agents with $|\mathcal{N}|=N$, and for each $i \in \mathcal{N}$, a strict preference ordering $P_i$ over the agents in $\mathcal{N} - \left\lbrace i \right\rbrace.$ We assume that the input preference orderings are derived from a set of underlying \textbf{hidden} edge weights ($w(x,y)$ for edge $x,y \in \mathcal{N}$), which satisfy the triangle inequality, i.e., for $x,y,z \in \mathcal{N}$, $w(x,y) \leq w(x,z) + w(y,z)$. These weights are considered to represent the \emph{ground truth}, which is not known to the algorithm. We say that the preferences $P$ are induced by weights $w$ if $\forall x,y,z \in \mathcal{N}$, if $x$ prefers $y$ to $z$, then $w(x,y) \geq w(x,z)$.
Our framework captures a number of well-motivated settings (for matching and clustering problems); we highlight two of them below.

\begin{enumerate}
\item\textbf{Forming Diverse Teams }
Our setting and objectives align with the research on \emph{diversity maximization algorithms}, a topic that has gained significant traction, particularly with respect to forming diverse teams that capture distinct perspectives~\cite{indykMMM14,marcolinoJT13}. In these problems, each agent corresponds to a point in a metric space: this point represents the agents's beliefs, skills, or opinions. Given this background, our matching problem essentially reduces to selecting diverse teams (of size two) based on different diversity goals, since points that are far apart ($w(x,y)$ is large) contribute more to the objective. For instance, one can imagine a teacher pairing up her students who possess differing skill sets or opinions for a class project, which is captured by the maximum weighted matching problem. In section~\ref{sec:applications}, we tackle the problem of forming diverse teams of arbitrary sizes by extending our model to encompass clustering, and team formation.


\item \textbf{Friendship Networks} In structural balance theory~\cite{davis1977clustering}, the statement that \emph{a friend of a friend is my friend} is folklore; this phenomenon is also exhibited by many real-life social networks~\cite{goodreau2009birds}. More generally, we can say that a graph with continuous weights has this property if $w(x,y) \geq \alpha[w(x,z) + w(y,z)]$ $\forall x,y,z$, for some suitably large $\alpha \leq \frac{1}{2}$. Friendship networks bear a close relationship to our model; in particular every graph that satisfies the friendship property for $\alpha \geq \frac{1}{3}$ must have metric weights, and thus falls within our framework.
\end{enumerate}

\noindent
In this paper our main goal is to form \textbf{ordinal approximation algorithms} for weighted matching problems, which we later extend towards other problems. An algorithm $A$ is said to be ordinal if
it only takes preference orderings $(P)_{i \in \mathcal{N}}$ as input (and not the hidden numerical weights $w$). It is an $\alpha$-approximation algorithm if for all possible weights $w$, and the corresponding induced preferences $P$, we have that $\frac{OPT(w)}{A(P)} \leq \alpha$. Here $OPT(w)$ is the total value of the maximum weight solution with respect to $w$, and $A(P)$ is the value of the solution returned by the algorithm for preferences $(P)_{i \in \mathcal{N}}$. In other words, such algorithms produce solutions which are always a factor $\alpha$ away from optimum, {\em without actually knowing} what the weights $w$ are.

In the rest of the paper, we focus primarily on the \textbf{Maximum Weighted Matching}(MWM) problem where the goal is to compute a matching to maximize the total (unknown) weight of the edges inside. A close variant of the MWM problem that we will also study is the Max $k$-Matching (M$k$-M) problem where the objective is to select a maximum weight matching consisting of at most $k$ edges. In addition, we also provide ordinal approximations for the following problems:

\begin{description}
\item [Max $k$-Sum] Given an integer $k$, partition the nodes into $k$ disjoint sets $(S_1, \ldots, S_k)$ of equal size in order to maximize $\sum_{i=1}^k \sum_{x,y \in S_i}w(x,y)$. (It is assumed that $N$ is divisible by $k$). When $k=N/2$, Max $k$-sum reduces to the Maximum Weighted Matching problem.

\item [Densest $k$-subgraph] Given an integer $k$, compute a set $S \subseteq \mathcal{N}$ of size $k$ to maximize the weight of the edges inside $S$.

\item [Max TSP] In the maximum traveling salesman problem, the objective is to compute a tour $T$ (cycle that visits each node in $\mathcal{N}$ exactly once) to maximize $\sum_{(x,y) \in T}w(x,y)$.
\end{description}


\subsubsection*{Challenges and Techniques}

We describe the challenges involved in designing ordinal algorithms through the lens of the Maximum Weighted Matching problem. First, different sets of edge weights may give rise to the same preference ordering and moreover, for each of these weights, the optimum matching can be different.
Therefore, unlike for the full information setting, no algorithm (deterministic or randomized) can compute the optimum matching using only ordinal information. More generally, the restriction that only ordinal information is available precludes almost all of the well-known algorithms for computing a matching. So, what kind of algorithms use only preference orderings? One algorithm which can still be implemented is a version of the extremely popular greedy matching algorithm, in which we successively select pairs of agents who choose each other as their top choice. Another trivial algorithm is to choose a matching at random: this certainly does not require any numerical information! It is not difficult to show that both these algorithms actually provide an ordinal $2$-approximation for the maximum weight matching. The main result of this paper, however, is that by interleaving these basic greedy and random techniques in non-trivial ways, it is actually possible to do much better, and obtain a $1.6$-approximation algorithm. Moreover, these techniques can further be extended and tailored to give ordinal approximation algorithms for much more general problems.

\subsection*{Our Contributions}

\begin{table}
\centering
\begin{tabular}{|l|c|c|l|l|}
\hline
\textbf{Problem} & \textbf{Full Info} & \multicolumn{2}{|c|}{\textbf{Our Results} (Ordinal Bounds)}\\
& & Deterministic & Randomized\\
\hline
Max Weighted Matching & $1$~\cite{edmonds1965paths} & $2$ & $1.6$ \\
\hline
Max $k$-Matching & $1$~\cite{lovasz2009matching} & $2$ & $2$ \\
\hline
Max $k$-Sum & $2$~\cite{feo1990class,hassin1997approximation} & $4$ & $3.2$  \\
\hline
Densest $k$-Subgraph & $2$~\cite{birnbaum2009improved,hassin1997approximation} &  $4$ & $4$ \\
\hline
Max TSP &$\frac{8}{7}$~\cite{kowalik2009deterministic} & $3.2$ & $2.14$\\
\hline
\end{tabular}
\caption{A Comparison of the approximation factors obtained by ordinal approximation algorithms and previous results for the full information metric setting. All of our results for the non-matching problems are obtained by using our algorithms (deterministic and randomized) for matching as a black-box to construct solutions for the respective problems.} 
\label{table_results}
\end{table}

Our main results are summarized in Table \ref{table_results}. All of the problems that we study have a rich history of algorithms for the full information setting. As seen in the table, our ordinal algorithms provide approximation factors that are close to the best known for the full information versions. In other words, we show that it is possible to find good solutions to such problems even without knowing any of the true weights, using only ordinal preference information instead.

Our central result in this paper is an ordinal $1.6$-approximation algorithm for max-weight matching; this is obtained by a careful interleaving of greedy and random matchings. We also present a deterministic $2$-approximation algorithm for Max $k$-Matching. Note that Max $k$-Matching for $k=\frac{N}{2}$ is the same as the MWM problem. 

We also provide a general way to use matching algorithms as a black-box to form ordinal approximation algorithms for other problems: given an ordinal $\alpha$-approximation for max-weight matching, we show how to obtain a $2\alpha$, $2\alpha$, and $\frac{4}{3}\alpha$ approximation for Max $k$-Sum, Densest $k$-Subgraph, and Max-TSP respectively. Plugging in the appropriate values of $\alpha$ for deterministic and randomized algorithms yields the results in Table~\ref{table_results}.

In total, our results indicate that for matching and clustering problems with metric preferences, ordinal algorithms perform almost as well as algorithms which know the underlying metric weights.

\noindent\textbf{Techniques:} More generally, one of our main contributions is a framework that allows the design of algorithms for problems where the (metric) weights are hidden. Our framework builds on two simple techniques, \emph{greedy} and \emph{random}, and establishes an interesting connection between graph density, matchings, and greedy edges. We believe that this framework may be useful for designing ordinal approximation algorithms in the future.

\subsection*{Related Work}
Broadly speaking, the cornucopia of algorithms proposed in the matching literature belong to one of two classes:  $(i)$ Ordinal algorithms that ignore agent utilities, and focus on (unquantifiable) axiomatic properties such as \emph{stability}, and $(ii)$ Optimization algorithms where the numerical utilities are fully specified. From our perspective, algorithms belonging to the former class, with the exception of \emph{Greedy}, do not result in good approximations for the hidden optimum, whereas the techniques used in the latter (e.g.,~\cite{drakeH03,duanP10}) depend heavily on \emph{improving cycles} and thus, are unsuitable for ordinal settings. A notable exception to the above dichotomy is the class of optimization problems studying \emph{ordinal measures of efficiency}~\cite{abrahamIKM07,chakrabartyS14}, for example, the average rank of an agent's partner in the matching.  Such settings often involve the definition of `new utility functions' based on given preferences, and thus are fundamentally different from our model where preexisting cardinal utilities give rise to ordinal preferences.

The idea of preference orders induced by metric weights (or a more general utility space) was first considered in the work of Irving et al.~\cite{irving1987efficient}. Subsequent work has focused mostly on analyzing the greedy algorithm or on settings where the agent utilities are explicitly known~\cite{arkinBEOMP09,emekLW15}. Most similar to our work is the recent paper by Filos-Ratsikas et al.~\cite{filosRatsikasF014}, who prove that for one-sided matchings, no ordinal algorithm can provide an approximation factor better than $\Theta({\sqrt{N}})$. In contrast, for two-sided matchings, there is a simple (greedy) $2$-approximation algorithm even when the hidden weights do not obey the metric inequality.

As with Matching, all of the problems studied in this paper have received considerable attention in the literature for the full information case with metric weights. In particular, metric Densest Subgraph (also known as \emph{Maximum Dispersion} or \emph{Remote Clique}) is quite popular owing to its innumerable applications~\cite{birnbaum2009improved,bhattacharyaGM11}. The close ties between the optimum solutions for Matching and Max $k$-sum, and Densest $k$-subgraph was first explored by Feo and Khellaf~\cite{feo1990class}, and later by Hassin et al.~\cite{hassin1997approximation}; our black-box mechanism to transform arbitrary matchings into solutions for other problems can be viewed as a generalization of their results. In addition, we also provide improved algorithms for these problems (see Table~\ref{table_results}) that do not depend on matchings; for Max $k$-sum, the bound that we obtain for the ordinal setting is as good as that of the best-known algorithm for the full information setting.

\noindent\textbf{Distortion in Social Choice} Our work is similar in motivation to the growing body of research studying settings where the voter preferences are induced by a set of hidden utilities~\cite{anshelevichBP15,boutilierCHLPS15,caragiannisP11,marcolino2014give,procacciaR06}. The voting protocols in these papers are essentially ordinal approximation algorithms, albeit for a very specific problem of selecting the utility-maximizing candidate from a set of alternatives.

Finally, other models of incomplete information have been considered in the Matching literature, most notably Online Algorithms~\cite{kalyanasundaramP93} and truthful algorithms (for strategic agents)~\cite{dughmiG10}. Given the strong motivations for preference rankings in settings with agents, it would be interesting to see whether algorithms developed for other partial information models can be extended to our setting.

\section{Framework for Ordinal Matching Algorithms}
\label{sec:framework}
In this section, we present our framework for developing ordinal approximation algorithms and establish tight upper and lower bounds on the performance of algorithms that select matching edges either greedily or uniformly at random. As a simple consequence of this framework, we show that the algorithms that sequentially pick all of the edges greedily or uniformly at random both provide $2$-approximations to the maximum weight matching. In the following section, we show how to improve this performance by picking some edges greedily, and some randomly. Finally, we remark that for the sake of convenience and brevity, we will often assume that $N$ is even, and sometimes that it is also divisible by 3. As we discuss in the Appendix, our results still hold if this is not the case, with only minor modifications.

\subsection{Fundamental Subroutine: Greedy }
We begin with Algorithm~\ref{alg_greedy} that describes a simple greedy procedure for outputting a matching: at each stage, the algorithm picks one edge $(x,y)$ such that the both $x$ and $y$ prefer this edge to all of the other available edges. We now develop some notation required to analyze this procedure.

\begin{definition}{(Undominated Edges)}
Given a set $E$ of edges, $(x,y) \in E$ is said to be an undominated edge if for all $(x,a)$ and $(y,b)$ in $E$, $w(x,y) \geq w(x,a)$ and $w(x,y) \geq w(y,b)$.
\end{definition}

\begin{algorithm}[htbp]
\SetKwInOut{Input}{input}\SetKwInOut{Output}{output}
\Input{Edge set $E$, preferences $P(\mathcal{N})$, $k$}
\Output{Matching $M_G$ with $k$ edges}
\While{$E$ is not empty (AND) $|M_G|<k$  }{
pick an undominated edge $e=(x,y)$ from $E$ and add it to $M_G$\;
remove all edges containing $x$ or $y$ from $E$\;
}
\caption{Greedy $k$-Matching Algorithm}
\label{alg_greedy}
\end{algorithm}

Given a set $E$, let us use the notation $E_\top$ to denote the set of undominated edges in $E$. Finally, we say that an edge set $E$ is complete if $\exists$ some $S \subseteq \mathcal{N}$ such that $E$ is the complete graph on the nodes in $S$ (minus the self-loops). We make the following two observations regarding undominated edges
\begin{enumerate}
\item Every edge set $E$ has at least one undominated edge. In particular, any maximum weight edge in $E$ is obviously an undominated edge.
\item Given an edge set $E$, one can efficiently find at least one edge in $E_\top$ {\em using only the ordinal preference information}. A naive algorithm for this is as follows. Consider starting with an arbitrary node $x$. Let $(x,y)$ be its first choice out of all the edges in $E$ (i.e., $y$ is $x$'s first choice of all the nodes it has an edge to in $E$). Now consider $y$'s first choice. If it is $x$, then the edge $(x,y)$ must be undominated, as desired. If instead it is some $z\neq x$, then continue this process with $z$. Eventually this process must cycle, giving us a cycle of nodes $x_0,x_1,\ldots,x_{\l-1}$ such that $x_i$ is the top preference of $x_{i-1}$, taken with respect to $\mod \l$. This means that all edges in this cycle have equal weight, even though we do not know what this weight is, since $x_i$ preferring $x_{i+1}$ over $x_{i-1}$ means that $w(x_i,x_{i+1})\geq w(x_i,x_{i-1})$. Moreover, the edge weights of all edges in this cycle must be the highest ones incident on the nodes in this cycle, since they are all top preferences of the nodes. Therefore, all edges in this cycle are undominated, as desired.
\end{enumerate}

In general, an edge set $E$ may have multiple undominated edges that are not part of a cycle. Our first lemma shows that these different edges are comparable in weight.

\begin{lemma}
\label{lem_undomweight}
Given a complete edge set $E$, the weight of any undominated edge is at least half as much as the weight of any other edge in $E$, i.e., if $e=(x,y) \in E_\top$, then for any $(a,b) \in E$, we have $w(x,y)\geq \frac{1}{2}w(a,b)$. This is true even if $(a,b)$ is another undominated edge.
\end{lemma}
\begin{proof}
Since $(x,y)$ is an undominated edge, and since $E$ is a complete edge set this means that $w(x,y) \geq w(x,a)$, and $w(x,y) \geq w(x,b)$. Now, from the triangle inequality, we get $w(a,b) \leq w(a,x) + w(b,x) \leq 2w(x,y).$ \end{proof}

It is not difficult to see that when $k=\frac{N}{2}$, the output of Algorithm~\ref{alg_greedy} coincides with that of the extremely popular greedy algorithm that picks the maximum weight edge at each iteration, and therefore, our algorithm yields an ordinal $2$-approximation for the MWM problem. Our next result shows that the approximation factor holds even for Max $k$-Matching, for any $k$: this is not a trivial result because at any given stage there may be multiple undominated edges and therefore for $k < \frac{N}{2}$, the output of Algorithm~\ref{alg_greedy} no longer coincides with that the well known greedy algorithm. In fact, we show the following much stronger lemma,

\begin{lemma}
\label{lem_matchinggreedy}
Given $k=\alpha\frac{N}{2}$, and $k^*=\alpha^* \frac{N}{2}$, the performance of the greedy $k$-matching with respect to the optimal $k^*$-matching (i.e., $\frac{OPT(k^*)}{Greedy(k)}$) is given by,
\begin{enumerate}
\item $\displaystyle\max{(2,2\frac{\alpha^*}{\alpha})}$ \quad if $\alpha^* + \alpha < 1$
\item $\displaystyle\max{(2, \frac{\alpha^* + 1}{\alpha} -1)}$  ~ if $\alpha^* + \alpha \geq 1$
\end{enumerate}
\end{lemma}

Thus, for example, when $\alpha^*=1$, and $\alpha=\frac{2}{3}$, we get the factor of $0.5$, i.e., in order to obtain a half-approximation to the optimum perfect matching, it suffices to greedily choose two-thirds as many edges as in the perfect matching.

\begin{proof}
We show the claim via a charging argument where every edge in the optimum matching $M^*$ is charged to one or more edges in the greedy matching $M$. Specifically, we can imagine that each edge $e \in M$ contains a certain (not necessarily integral) number of slots $s_e$, initialized to zero, that measure the number of edges in $M^*$ charged to $e$. Our proof will proceed in the form of an algorithm: initially $U=M^*$ denotes the set of uncharged edges. In each iteration, we remove some edge from $U$, charge its weight to some edges in $M$ and increase the value of $s_e$ for the corresponding edges so that the following invariant always holds: $\sum_{e \in M}s_ew_e \geq \sum_{e^* \in M^* \setminus U} w_{e^*}.$ Finally, we can bound the performance ratio using the quantity $\max_{e \in M}s_e$.

We describe our charging algorithm in three phases. Before we describe the first phase, consider any edge $e^*=(a,b)$ in $M^*$. The edge must belong to one of the following two types.

\begin{enumerate}
\item (Type I) Some edge(s) consisting of $a$ or $b$ (both $a$ and $b$) are present in $M$.
\item (Type II) No edge in $M$ has $a$ or $b$ as an endpoint.
\end{enumerate}

Suppose that $M^*$ contains $m_1$ Type I edges, and $m_2$ Type II edges. We know that $m_1 + m_2 = k^*$. Also, let $T \subseteq M$ denote the top $\frac{m_1}{2}$ edges in $M$, i.e., the $\frac{m_1}{2}$ edges with the highest weight. In the first charging phase, we cover all the Type I edges using only the edges in $T$, and so that no more than two slots of each edge are required.

\begin{claim}
\label{clm_greedyfirstphase}
(First Phase) There exists a mechanism by which we can charge all Type I edges in $M^*$ to the edges in $T$ so that $\sum_{e \in T}s_e w_e \geq \sum_{e \in Type I(M^*)}w_e$ and for all $e \in T$, $s_e \leq 2$.
\end{claim}
\begin{proof}
We begin by charging the Type I edges to arbitrary edges in $M$, and then transfer the slots that are outside $T$ to edges in $T$. Consider any Type I edge $e^*=(a,b)$: without loss of generality, suppose that $e=(a,c)$ is the first edge containing either $a$ or $b$ that was added to $M$ by the greedy algorithm. Since the greedy algorithm only adds undominated edges, we can infer that $w_e \geq w_{e^*}$ (or else $e$ would be dominated by $e^*$). Using this idea, we we charge the Type I edges as follows
\begin{quote}
(Algorithm: Phase I (Charging)) Repeat until $U$ contains no Type I edge: pick a type I edge $e^*=(a,b)$ from $U$. Suppose that $e=(a,c)$ is the first edge containing either $a$ or $b$ that was added to $M$. Since $w_e \geq w_{e^*}$, charge $e^*$ to $e$, i.e., increase $s_e$ by one and remove $e^*$ from $U$.
\end{quote}

At the end of the above algorithm, $U$ contains no type $I$ edge. Moreover, $\sum_{e \in M}s_e = m_1$ since every Type I edge requires only one slot. Finally, for every $e=(x,y) \in M$, $s_e \leq 2$. This is because any edge charged to $(x,y)$ must contain at least one of $x$ or $y$. Now, without altering the set of uncharged edges $U$, we provide a mechanism to transfer the slots to edges in $T$. The following procedure is based on the observation that for every $e, e' \in M$ such that $e' \in T$ and $e \notin T$, $w_{e'} \geq w_{e}$.

\begin{quote}
(Algorithm: Phase I (Slot Transfer)) Repeat until $s_e=0$ for every edge outside $T$: pick $e \notin T$ such that $s_e > 0$. Pick any edge $e' \in T$ such that $s_e < 2$. Transfer the edge originally charged to $e$ to $e'$, i.e., decrease $s_e$ by one and increase $s_{e'}$ by one.
\end{quote}

Notice that at the end of the above mechanism, $\sum_{e \in T}s_e = m_1$, $s_e \leq 2$ for all $e \in T$, and $s_e =0$ for all $e \in M\setminus T$.
\end{proof}

Now, consider any type II edge $e^*$. We make a strong claim: for every $e \in M$, $w_e \geq \frac{1}{2}w_{e^*}$. This follows from Lemma~\ref{lem_undomweight} since at the instant when $e$ was added to $M$, $e$ was an undominated edge in the edge set $E$ and $e^*$ was also present in the edge set. Therefore, each type II edge can be charged using two (unit) slots from any of the edges in $M$ (or any combination of them). We now describe the second phase of our charging algorithm that charges nodes only to edges in $M \setminus T$, recall that there $k - \frac{m_1}{2}$ such edges.

\begin{quote}
(Second Phase) Repeat until $s_e=2$ for all $e \in M \setminus T$ (or) until $U$ is empty: pick any arbitrary edge $e^*$ from $U$ and $e \in M \setminus T$ such that $s_e = 0$. Since $w_{e^*} \leq 2w_e$, charge $e^*$ using two slots of $e$, i.e., increase $s_e$ by two and remove $e^*$ from $U$.
\end{quote}

During the second phase, every edge in $M^*$ is charged to exactly (two slots of) one edge in $M \setminus T$. Therefore, the number of edges removed from $U$ during this phase is $\min{(m_2, k - \frac{m_1}{2})}$. Since the number of uncharged edges at the beginning of Phase I was exactly $m_2$, we conclude that the number of uncharged edges at the end of the second phase, i.e., $|U|$ is $\min{(0,m_2 - k + \frac{m_1}{2})}$. If $|U|=0$, we are done, otherwise we can charge the remaining edges in $U$ uniformly to all the edges in $M$ using a fractional number of slots, i.e.,

\begin{quote}
(Third Phase) Repeat until $U=\emptyset$: pick any arbitrary edge $e^*$ from $U$. Since $w_{e^*} \leq 2w_e$ for all $e\in M$, charge $e^*$ uniformly to all edges in $M$, i.e., increase $s_e$ by $\frac{2}{k}$ for every $e \in M$ and remove $e^*$ from $M^*$.
\end{quote}

Now, in order to complete our analysis, we need to obtain an upper bound for $s_e$ over all edges in $e$. Recall that at the end of phase II, $s_e \leq 2$ for all $e \in M$. In the third phase, $s_e$ increased by $\frac{2}{k}$ for every edge in $U$, and the number of edges in $U$ is $\min{(0,m_2 - k + \frac{m_1}{2})}$. Therefore, at the end of the third phase, we have that for every $e \in M$,

$$s_e \leq 2 + \frac{2}{k}\displaystyle[\min{(0,m_2 - k + \frac{m_1}{2})}.]$$

Since $m_1 + m_2 = k^*$, we can simplify the second term above and get
\begin{align}
\label{eqn_temp1}
s_e & \leq 2 + \frac{2}{k}\displaystyle[\min{(0,\frac{m_2}{2} + \frac{k^*}{2} - k})]\\
& =  2 + \min(0,\frac{m_2 + k^*}{k} - 2)] = \min(2,\frac{m_2+k^*}{k}).
\end{align}

How large can $m_2$ be? Clearly, $m_2 \leq k^*$. But a more careful bound can be obtained using the fact that the $m_2$ Type II edges have no node in common with any of the $k$ edges in $M$. But the total number of nodes is $N$, therefore, $2m_2 + 2k \leq N$ or $m_2 \leq \frac{N}{2} - k$. This gives us $m_2 \leq \min(k^*, \frac{N}{2}-k)$. Depending on what the minimum is, we get two cases:
\begin{enumerate}
\item Case I: $k^* \leq \frac{N}{2}-k$ or equivalently, $k^* + k \leq \frac{N}{2}$. Substituting $m_2 \leq k^*$ in Equation~\ref{eqn_temp1}, we get that for all $e$, $s_e \leq \min(2, \frac{2k^*}{k})$. Replacing $k$ by $\alpha\frac{N}{2}$ and $k^*$ by $\alpha^* \frac{N}{2}$, we get that when $\alpha + \alpha^* \leq 1$, $s_e \leq \min(2,2\frac{\alpha^*}{\alpha})$.

\item Case II: $k^* \geq \frac{N}{2}-k$ or equivalently $\alpha^* + \alpha \geq 1$. Substituting $m_2 \leq \frac{N}{2}-k$ in Equation~\ref{eqn_temp1}, we get that $s_e \leq \displaystyle\min(2,\frac{\frac{N}{2} + k^* -k}{k})$ or equivalently $s_e \leq \min(2, \frac{\alpha^*+1}{\alpha}-1)$.
\end{enumerate}
\end{proof}
Plugging in $k=k^*$ in the above lemma immediately gives us the following corollary.
\begin{corollary}
Algorithm~\ref{alg_greedy} is a deterministic, ordinal $2$-approximation algorithm for the Max $k$-Matching problem for all $k$, and therefore a $2$-approximation algorithm for the Maximum Weighted Matching problem.
\end{corollary}

\subsection*{Fundamental Subroutine: Random}
An even simpler matching algorithm is simply to form a matching completely at random; this does not even depend on the input preferences. This is formally described in Algorithm~\ref{alg_random}. In what follows, we show upper and lower bounds on the performance of Algorithm~\ref{alg_greedy} for different edges sets.
\begin{algorithm}[htbp]
\SetKwInOut{Input}{input}\SetKwInOut{Output}{output}
\Input{Edge set $E$, $k$}
\Output{Matching $M_R$ with $k$ edges}
\While{$E$ is not empty (AND) $|M_R| < k$}{
pick an edge from $E$ uniformly at random. Add this edge $e=(x,y)$ to $M_R$\;
remove all edges containing $x$ or $y$ from $E$\;
}
\caption{Random $k$-Matching Algorithm}
\label{alg_random}
\end{algorithm}

\begin{lemma}
\label{lem_randomlower}
(Lower Bound)
\begin{enumerate}

\item Suppose $G=(T,E)$ is a complete graph on the set of nodes $T \subseteq \mathcal{N}$ with $|T|=n$. Then, the expected weight of the random (perfect) matching returned by Algorithm~\ref{alg_random} for the input $E$ is
$E[w(M_R)] \geq \frac{1}{n}\sum_{(x,y) \in E} w(x,y).$

\item Suppose $G=(T_1,T_2,E)$ is a complete bipartite graph on the set of nodes $T_1, T_2 \subseteq \mathcal{N}$ with $|T_1|=|T_2|=n$. Then, the weight of the random (perfect) matching returned by Algorithm~\ref{alg_random} for the input $E$ is
$E[w(M_R)] = \frac{1}{n}\sum_{(x,y) \in E} w(x,y).$
\end{enumerate}
\end{lemma}
\begin{proof}
We show both parts of the theorem using simple symmetry arguments. For the complete (non-bipartite) graph, let $\mathcal{M}$ be the set of all perfect matchings in $E$. Then, we argue that every matching $M$ in $\mathcal{M}$ is equally likely to occur. Therefore, the expected weight of $M_R$ is

\begin{equation}
\label{eqn_expectation}
E[w(M_R)] = \frac{1}{|\mathcal{M}|}\sum_{M \in \mathcal{M}} w(M) = \sum_{e=(x,y) \in E}p_e w(x,y),
\end{equation}

where $p_e$ is the probability of edge $e$ occurring in the matching. Since the edges are chosen uniformly at random,  the probability that a given edge is present in $M_R$ is the same for all edges in $E$. So $\forall e$, we have the following bound of $p_e$, which we can substitute in Equation~\ref{eqn_expectation} to get the first result. 
$$p_e = \frac{|M_R|}{|E|} = \frac{n/2}{n(n-1)/2} = \frac{1}{n-1}\geq \frac{1}{n},$$ 

For the second case, where $E$ is the set of edges in a complete bipartite graph, it is not hard to see that once again every edge $e$ is present in the final matching with equal probability. Therefore, $p_e = \frac{|M_R|}{|E|} = \frac{n}{n^2} = \frac{1}{n}.$
\end{proof}

\begin{lemma}
\label{lem_matchingupper}
(Upper Bound) Let $G=(T,E)$ be a complete subgraph on the set of nodes $T \subseteq \mathcal{S}$ with $|T|=n$, and let $M$ be any perfect matching on the larger set $\mathcal{S}$. Then, the following is an upper bound on the weight of $M$,
$$w(M) \leq \frac{2}{n} \sum_{\substack{x \in T \\ y \in T}} w(x,y) + \frac{1}{n}\sum_{\substack{x \in T \\ y \in \mathcal{S} \setminus T}}w(x,y) $$
\end{lemma}
\begin{proof}
Fix an edge $e=(x,y) \in M$. Then, by the triangle inequality, the following must hold for every node $z \in T$: $w(x,z)+w(y,z) \geq w(x,y)$. Summing this up over all $z \in T$, we get
$$\sum_{z \in T}w(x,z) + w(y,z) \geq n w(x,y) = n(w_e).$$

Once again, repeating the above process over all $e \in M$, and then all $z \in T$ we have
$$nw(M) \leq 2 \sum_{\substack{x \in T \\ y \in T}} w(x,y) + \sum_{\substack{x \in T \\ y \in \mathcal{S} \setminus T}}w(x,y) $$

Each $(x,y) \in E$ appears twice in the RHS: once when we consider the edge in $M$ containing $x$, and once when we consider the edge with $y$.
\end{proof}

We conclude by proving that picking edges uniformly at random yields a $2$-approximation for the MWM problem.

\begin{claim}
Algorithm~\ref{alg_random} is an ordinal $2$-approximation algorithm for the Maximum Weighted Matching problem.
\end{claim}
\begin{proof}
From Lemma~\ref{lem_randomlower}, we know that in expectation, the matching output by the algorithm when the input is $\mathcal{N}$ has a weight of at least $\frac{1}{N}\sum_{x \in \mathcal{N}, y \in \mathcal{N}}w(x,y)$. Substituing $T = \mathcal{S} = \mathcal{N}$ in Lemma~\ref{lem_matchingupper} and $M=OPT$ (max-weight matching) gives us the following upper bound on the weight of $OPT$, $w(OPT) \leq \frac{2}{N} \sum_{x \in \mathcal{N}, y \in \mathcal{N}}w(x,y) \leq 2 E[w(M_R)]$.
\end{proof}

\subsection{Lower Bound Example for Ordinal Matchings}
Before presenting our algorithms, it is important to understand the limitations of settings with ordinal information. As mentioned in the Introduction, different sets of weights can give rise to the same preference ordering, and therefore, we cannot suitably approximate the optimum solution for every possible weight. We now show that even for very simple instances, there can be no deterministic $1.5$-approximation algorithm, and no randomized $1.25$-approximation algorithm. 

\begin{claim}
\label{clm:lowerbound}
No deterministic ordinal approximation algorithm can provide an approximation factor better than $1.5$, and no randomized ordinal approximation algorithm can provide an approximation factor better than $1.25$ for Maximum Weighted Matching. No ordinal algorithm, deterministic or randomized can provide an approximation factor better than $2$ for the Max $k$-Matching problem.
\end{claim}
\begin{proof}
Consider an instance with $4$ nodes having the following preferences: $(i)$ $a: b>c>d$, $(ii)$ $b: a>d>c$, $(iii)$ $c: a>b>d$, $(iv)$ $d: b>a>c$. Since the matching $\left\lbrace(a,d), (b,c)\right\rbrace$ is weakly dominated, it suffices to consider algorithms that randomize between $M_1 = \left\lbrace(a,b), (c,d)\right\rbrace$, and $M_2 =\left\lbrace(a,c), (b,d)\right\rbrace$, or deterministically chooses one of them.

Now, consider the following two sets of weights, both of which induce the above preferences but whose optima are $M_2$ and $M_1$ respectively: $W_1:= w(a,b)=w(a,c)=w(b,d)$ $=w(a,d)=w(b,c)=1$, $w(c,d)=\epsilon$, and $W_2:= w(a,b)=2, w(a,c)=w(b,d)=w(c,d)=w(a,d)=w(b,c)=1$. The best deterministic algorithm always chooses the matching $M_2$, but for the weights $W_2$, this is only a $\frac{3}{2}$-approximation to OPT.

Consider any randomized algorithm that chooses $M_1$ with probability $x$, and $M_2$ with probability $(1-x)$. With a little algebra, we can verify that just for $W_1$, and $W_2$, the optimum randomized algorithm has $x=\frac{2}{5}$, yielding an approximation factor of $1.25$.

For the Max $k$-Matching problem, our results are tight. For small values of $k$, it is impossible for any ordinal algorithm to provide a better than $2$-approximation factor. To see why, consider an instance with $2N$ nodes $\{a_1,b_1,a_2,b_2,\ldots, a_N,b_N\}$. Every $a_i$'s first choice is $b_i$ and vice-versa, the other preferences can be arbitrary. Pick some $i$ uniformly at random and set $w(a_i,b_i)=2$, and all the other weights are equal to $1$. For $k=1$, it is easy to see that no randomized algorithm can always pick the max-weight edge and therefore, as $N \to \infty$, we get a lower bound of $2$.
\end{proof}

Since Max $k$-sum is a strict generalization of Maximum Weighted Matching, the same lower bounds for Maximum Matching hold for Max $k$-sum as well.


\section{Ordinal Matching Algorithms}
Here we present a better ordinal approximation than simply taking the random or greedy matching. The algorithm first performs the greedy subroutine until it matches $\frac{2}{3}$ of the agents. Then it either creates a random matching on the unmatched agents, {\em or} it creates a random matching between the unmatched agents and a subset of agents which are already matched. We show that one of these matchings is guaranteed to be close to optimum in weight. Unfortunately since we have no access to the weights themselves, we cannot simply choose the best of these two matchings, and thus are forced to randomly select one, giving us good performance in expectation. More formally, the algorithm is:

\begin{algorithm}[htbp]
\SetKwInOut{Input}{input}\SetKwInOut{Output}{output}
\Input{$\mathcal{N}, P(\mathcal{N})$}
\Output{Perfect Matching $M$}
Initialize $E$ to the be complete graph on $\mathcal{N}$, and $M_1=M_2=\emptyset$\;
Let $M_0$ be the output returned by Algorithm~\ref{alg_greedy} for $E$,$k=\frac{2}{3}\frac{N}{2}$\;
Let $B$ be the set of nodes in $\mathcal{N}$ not matched in $M_0$, and $E_B$ is the complete graph on $B$.\;
\textbf{First Algorithm}\;
$M_1 = M_0 \cup (\text{The perfect matching output by Algorithm~\ref{alg_random} on $E_B$)}$\;
\textbf{Second Algorithm} \;
Choose half the edges from $M_0$ uniformly at random and add them to $M_2$\;
Let $A$ be the set of nodes in $M_0 \setminus M_2$\;
Let $E_{ab}$ be the edges of the complete bipartite graph $(A,B)$\;
Run Algorithm~\ref{alg_random} on the set of edges in $E_{ab}$ to obtain a perfect bipartite matching and add the edges returned by the algorithm to $M_2$\;
\textbf{Final Output}
Return $M_1$ with probability $0.5$ and $M_2$ with probability $0.5$.
\caption{$1.6$-Approximation Algorithm for Maximum Weight Matching}
\label{alg_bestmatching}
\end{algorithm}

\begin{theorem}
\label{thm_mainmatching}
For every input ranking, Algorithm~\ref{alg_bestmatching} returns a $\frac{8}{5}=1.6$-approximation to the maximum-weight matching.
\end{theorem}

\begin{proof}
First, we provide some high-level intuition on why this algorithm results in a significant improvement over the standard half-optimal greedy and randomized approaches. Observe that in order to obtain a half-approximation to $OPT$, it is sufficient to greedily select $\frac{2}{3}(N/2)$ edges (substitute $\alpha=\frac{2}{3}$, $\alpha^*=1$ in Lemma~\ref{lem_matchinggreedy}). Choosing all $\frac{N}{2}$ edges greedily would be overkill, and so we choose the remaining edges randomly in the First Algorithm of Alg~\ref{alg_bestmatching}. Now, let us denote by $Top$, the set of $\frac{2}{3}N$ nodes that are matched greedily. The main idea behind the second Algorithm is that if the first one performs poorly (not that much better than half), then, all the `good edges' must be going across the cut from $Top$ to Bottom ($B$). In other words, $\sum_{(x,y) \in Top \times B} w(x,y)$ must be large, and therefore, the randomized algorithm for bipartite graphs should perform well. In summary, since we randomized between the first and second algorithms, we are guaranteed that at least one of them should have a good performance for any given instance.

We now prove the theorem formally. By linearity of expectation, $E[w(M)] = 0.5(E[w(M_1)] + E[w(M_2)]).$ Now, look at the first algorithm, since $M_0$ has two-thirds as many edges as the optimum matching, we get from Lemma~\ref{lem_matchinggreedy} that $w(M_0) \geq \frac{1}{2}w(OPT)$. As mentioned in the algorithm, $B$ is the set of nodes that are not present in $M_0$; since we randomly match the nodes in $B$ to other nodes in $B$, the expected weight of the random algorithm (from Lemma~\ref{lem_randomlower} with $n=\frac{N}{3}$) is $\frac{3}{N}\sum_{(x,y) \in B}w(x,y).$ Therefore, we get the following lower bound on the weight of $M_1$,
$$E[w(M_1)] \geq \frac{OPT}{2} + \frac{3}{N}\sum_{(x,y) \in B}w(x,y).$$

Next, look at the second algorithm: half the edges from $M_0$ are added to $M_2$. $A$ constitutes the set of $\frac{N}{3}$ nodes from $M_0$ that are not present in $M_2$, these nodes are randomly matched to those in $B$. Let $M_{AB}$ denote the matching going `across the cut' from $A$ to $B$. Since the set $A$ is chosen randomly from the nodes in $Top$, the expected weight of the matching from $A$ to $B$ is given by,
\begin{align*}
E[w(M_{AB})]= &\sum_{\substack{S \subset Top \\ |S| = \frac{N}{3}}} E\left[w(M_{AB}) \mid (A=S)\right]Pr(A=S)\\
= & \sum_{\substack{S \subset Top \\ |S| = \frac{N}{3}}} Pr(A=S)\sum_{(x,y) \in S \times B} \frac{3}{N}w(x,y)\\
= & \sum_{(x,y) \in Top \times B} \frac{3}{N}w(x,y) Pr(x \in A)\\
= & \sum_{(x,y) \in Top \times B} \frac{3}{N}w(x,y) \times \frac{1}{2} & .
\end{align*}

The second equation above comes from Lemma~\ref{lem_randomlower} Part 2 for $n=\frac{N}{3}$, and the last step follows from the observation that $Pr(x \in A)$ is exactly equal to the probability that the edge containing $x$ in $M_0$ is not added to $M_2$, which is one half (since the edge is chosen with probability $0.5$). We can now bound the performance of $M_2$ as follows,
\begin{align*}
E[w(M_2)] & = \frac{1}{2}E[w(M_0)] + E[w(M_{AB})]\\
& = \frac{1}{2}w(M_0) + \frac{3}{2N}\sum_{(x,y) \in Top \times B}w(x,y).
\end{align*}

Now, let us apply Lemma~\ref{lem_matchingupper} to the set $T=B$ ($n=\frac{N}{3}$), with $OPT$ being the matching: we get that $w(OPT) \leq \frac{6}{N}\sum_{x \in B, y \in B}w(x,y) + \frac{3}{N}\sum_{x \in Top, y \in B}w(x,y)$ or equivalently $\frac{3}{2N}\sum_{(x,y) \in Top \times B}w(x,y) \geq \frac{w(OPT)}{2} - \frac{3}{N}\sum_{x \in B, y \in B}w(x,y).$ Substituting this in the above equation for $M_2$ along with the fact that $w(M_0) \geq \frac{w(OPT)}{2}$, we get the following lower bound for the performance of $M_2$ in terms of OPT,
$$E[w(M_2)] \geq \frac{w(OPT)}{4} + \frac{w(OPT)}{2} - \frac{3}{N}\sum_{x \in B, y \in B}w(x,y).$$
Recall that
$$E[w(M_1)] \geq \frac{w(OPT)}{2} + \frac{3}{N}\sum_{x \in B, y \in B}w(x,y).$$
The final bound comes from adding the two quantities above and multiplying by half.
\end{proof}

\subsection{Matching without the Metric Assumption}
We now discuss the general case where the hidden weights do not obey the triangle inequality. From our discussion in Section~\ref{sec:framework}, we infer that Algorithm~\ref{alg_greedy} still yields a $2$-approximation to the MWM problem as its output coincides with that of the classic greedy algorithm. No deterministic algorithm can provide a better approximation; consider the same preference orderings as Claim~\ref{clm:lowerbound} and the following two sets of consistent weights: $(i)$ $w(c,d)=\epsilon$, other weights are $1$, and $(ii)$ $w(a,b)=1$, other weights are $\epsilon$. The only good choice for case $(ii)$ is the matching $M_1$, which yields a $2$-approximation for case $(i)$.

We now go one step further and show that even if we are allowed to utilize randomized mechanisms, we still can't do that much better than an ordinal $2$-approximation factor. 

\begin{claim}
When the hidden weights do not obey the triangle inequality, there exist a set of preferences such that no randomized algorithm can provide an ordinal approximation factor better than $\frac{5}{3}$ for every set of weights consistent with these preferences.
\end{claim}
\begin{proof}
The instance consists of $8$ nodes $(a_i,b_i)_{i=1}^4$, and we describe the preferences in a slightly unconventional but more intuitive manner. Every node $p \in \mathcal{N}$ is assigned a rank $r(p)$ and two or more nodes may have the same rank. Now, the preference ordering for a node $q$ is simply a list of the nodes in $\mathcal{N} - \{q\}$ sorted in the ascending order of their rank and ties in the rank can be broken arbitrarily. For this instance, we set $rank(a_i) = rank(b_i) = i$ for $i=1$ to $4$. So for instance, $a_1$'s preference ordering can be: $b_1 > a_2 > b_2 > a_3 > b_3 > a_4 > b_4$.

For all of the weights that we consider, $w(a_1,b_1)=1$, $w(p,q)=0$ if $p=\{a_3,a_4\}$, $q=\{b_3,b_4\}$. Moreover, for a given $i$, the weight of all edges going out of $a_i$ and $b_i$ are the same. By symmetry, there exists an optimal randomized mechanism for this instance that randomizes among the following six matchings, choosing matching $M_i$ with probability $x_i$. 

\begin{enumerate}
\item $M_1 = \{(a_1,b_1), (a_2,b_2), (a_3,b_3), (a_4,b_4)\}$
\item $M_2 = \{(a_1,b_1), (a_2,a_3), (b_2,b_3), (a_4,b_4)\}$
\item $M_3 = \{(a_1,a_2), (b_1,b_2), (a_3,b_3), (a_4,b_4)\}$
\item $M_4 = \{(a_1,a_2), (b_1,a_3), (b_2,b_3), (a_4,b_4)\}$
\item $M_5 = \{(a_1,a_3), (b_1,b_3), (a_2,a_4), (b_2,b_4)\}$
\item $M_6 = \{(a_1,a_3), (b_1,b_3), (a_2,b_2), (a_4,b_4)\}$.
\end{enumerate}

We now construct sets of weights consistent with the preferences. Suppose that the optimal randomized strategy $A$ provides an ordinal approximation factor of $\frac{1}{c}$ for some $c \in (0,1]$. Then for every set of weights $W$, it is true that 

\begin{equation}
\label{eqn_ordinafactor}
\frac{OPT(W)}{E[A(w)]} \leq \frac{1}{c}
\end{equation}
 We now explicitly construct some weights and derive an upper bound of $\frac{3}{5}$ on $c$, which implies the no randomized strategy can have an approximation factor better than $\frac{5}{3}$ for this instance.

\begin{enumerate}
\item All weights are zero except $w(a_1,b_1)=1$. For this instance $OPT =1$, and only $M_1,M_2$ give non-zero utility, so $A(W) = x_1 + x_2$. Applying Equation~\ref{eqn_ordinafactor}, we get $x_1 + x_2 \geq c$.

\item $w(a_1,a_2) = w(a_1,b_2) = w(b_1,a_2) = w(b_1,b_2) = 1$, rest are zero. $OPT = 2$, $A(w) = x_1 + x_2 + 2x_3 + x_4$. The corrresponding inequality is $x_1 + x_2 + 2x_3 + x_4 \geq 2c$.

\item All weights going out of $a_1, b_1$ are one. Among the remaining weights, only $w(a_2,b_2) = 1$. $OPT = 3$. $A(W) = 2x_1 + x_2 + 2x_3 + 2x_4 + 2x_5 + 3x_6$, giving us the inequality, $2x_1 + x_2 + 2x_3 + 2x_4 + 2x_5 + 3x_6 \geq 3c$.

\item The final instance has all weights coming out of $a_1,a_2,b_1,b_2$ to be one. $OPT = 4$. The final inequality is $A(W) = 2x_1 + 3x_2 + 2x_3 + 3x_4 + 4x_5 + 3x_6 \geq 4c$.
\end{enumerate}
Adding all of the $4$ inequalities above gives us $\sum x_i \geq \frac{10}{6}c$. Since $\sum x_i = 1$, we get that $c \leq \frac{3}{5}$, which completes the proof.
\end{proof}

We hypothesize that as we extend the instance in the above claim, as $N \to \infty$, we should obtain a lower bound of $2$, which meets our upper bound. For Max $k$-matching, the situation is much more bleak; no algorithm, deterministic or randomized can provide a reasonable approximation factor if $k$ is small. As we did before, consider an instance with $2N$ nodes $\{a_1,b_1,a_2,b_2,\ldots, a_N,b_N\}$. Every $a_i$'s first choice is $b_i$ and vice-versa, the other preferences can be arbitrary. Pick some $i$ uniformly at random and set $w(a_i,b_i)=1$, and all the other weights are equal to $\epsilon$. For $k=1$, it is easy to see that every randomized algorithm obtains non-zero utility only with probability $\frac{1}{N}$, whereas $OPT=1$. Therefore, the ordinal approximation factor for any random algorithm is $N$ and as $N \to \infty$, the factor becomes unbounded. Moreover, for other values of $k$, the ordinal approximation factor for the same instance is $\frac{N}{k}$.

\section{Matching as a Black-Box for other Problems}
\label{sec:applications}
In this section, we highlight the versatility of matchings by showing how matching algorithms can be used as a black-box to obtain good (ordinal) approximation algorithms for other problems. These reductions serve as stand-alone results, as the algorithms for matching are easy to implement as well as extremely common in settings with preference lists. Moreover, future improvements on the ordinal approximation factor for matchings can be directly plugged in to obtain better bounds for these problems.

\subsubsection*{Informal Statement of Results}

\begin{enumerate}
\item \textbf{Max $k$-Sum}: Given an $\alpha$-approximate perfect matching $M$, we can obtain a nice clustering as follows: ``simply divide $M$ into $k$ equal sized sets (with $\frac{N}{2k}$ edges in each) and form clusters using the nodes in each of the equal-sized sets." It turns out that this simple mechanism provides a $2\alpha$-approximation to the optimum clustering. Plugging in our $1.6$-approximation algorithm, we immediately get a $3.2$-approximation algorithm for Max $k$-sum. 

\item \textbf{Densest $k$-Subgraph}: Suppose we are provided an $\alpha$-approximate matching $M$ of size $\frac{k}{2}$, how good is the set $S$ containing the $k$ nodes in $M$? Using Lemma~\ref{lem_matchingupper}, we can establish that $S$ is at least as dense as $\frac{w(M)}{2}$, and the density of the optimum solution is at most $\alpha w(M)$. Therefore, $S$ is a $2\alpha$-approximation to the optimum set of size $k$. This easy-to-implement mechanism directly yields a $4$-approximation algorithm for Densest $k$-subgraph.

\item \textbf{Max TSP:} Given an $\alpha$-approximate perfect matching $M$, any tour $T$ containing $M$ is a $2\alpha$-approximation since the weight of the optimum tour cannot be more than twice that of the optimum matching. However, if we carefully form $T$ using only undominated edges, we can show that the resulting solution is a $\frac{4}{3}$-approximation to the optimum tour. Plugging in $\alpha=1.6$, we get an ordinal $2.14$-approximation algorithm for Max TSP.
\end{enumerate}

\subsubsection*{Formal Results}
\begin{theorem}
\begin{enumerate}
\item Any $\alpha$-approximation algorithm for the Maximum Weight Perfect Matching problem can be used to obtain a $2\alpha$-approximation for the Max $k$-Sum problem.
\item Any $\alpha$-approximation algorithm for the Maximum Weight $\frac{k}{2}$-Matching problem can be used to obtain a $2\alpha$-approximation for the Densest $k$-Subgraph problem.
\item Any $\alpha$-approximation algorithm for the Maximum Weight Perfect Matching problem can be used to obtain a $\frac{4}{3 - \frac{4}{N}}\alpha$-approximation for the Max TSP problem. \end{enumerate}
\end{theorem}

\begin{proof}
\textbf{(Part 1)} Suppose that we are provided a perfect matching $M$ that is a $\alpha$-approximation to the optimum matching. We use the following simple procedure to cluster the nodes into $k$-clusters:
\begin{enumerate}
\item Initialize $k$ empty clusters $S_1, S_2, \ldots, S_k$, and $M'=M$.
\item While $\exists$ some $S_i$ such that $|S_i| \neq \frac{N}{k}$
\item Pick some edge $e \in M'$, add both the end-points of $e$ to $S_i$, remove $e$ from $M'$.
\end{enumerate}
The only important property we require is that for every edge $e \in M$, both its end points belong to the same cluster. We now prove that this a $2\alpha$-approximation to the optimum max $k$-sum solution ($OPT = (O_1, \ldots, O_k)$). Let $M^*$ be the optimum perfect matching, and let $c=\frac{N}{k}$. Finally, since the end points of every edge in $M$ belong to the same cluster, without l.o.g, let $M_i$ denote the edges of $M$ that are present in the cluster $S_i$. First, we establish a lower bound on the quality of our solution $S$,

\begin{align*}w(S) & = \sum_{i=1}^k \sum_{(x,y) \in S_i} w(x,y) \\
& \geq \sum_{i=1}^k \frac{c}{2} w(M_k) & \text{(Lemma~\ref{lem_matchingupper})} \\
& = \frac{1}{2} cw(M).
\end{align*}

Now, we establish an upper bound for $OPT$ in terms of $w(M)$. Suppose that $M^*(O_i)$ is the maximum weight perfect matching on the set of nodes $O_i$. Then,
\begin{align*}
w(OPT) & = \sum_{i=1}^k \sum_{(x,y) \in O_i}w(x,y) \\
& \leq \sum_{i=1}^k c w(M^*(O_i)) & \text{(Corollary~\ref{corr_optlower})}\\
& \leq c w(\bigcup_{i=1}^k M^*(O_i)) \leq c w(M^*) &\\
& \leq c\alpha w(M).
\end{align*}
Reconciling the two bounds gives us the desired factor of $2\alpha$.\\

\textbf{(Part 2)} Once again, we use $M^*$ to denote the optimum $\frac{k}{2}$-matching and $M$ to denote the $\alpha$-approximation. Let $O$ be the optimum solution to the Densest $k$-subgraph problem for the given value of $k$. Then our algorithm simply returns the solution $S$ compromising of the endpoints of all the edges in $M$. The proof is quite similar to the proof for Part 1.

\begin{align*}
w(O) & = \sum_{(x,y) \in O}w(x,y) \leq kw(M^*(O)) \leq kw(M^*)\\
& \leq \alpha k w(M) \leq \alpha (2\sum_{(x,y) \in S}w(x,y)).
\end{align*}

\textbf{(Part 3)} Let $OPT$ denote the optimum solution to the Max TSP problem, and $M^*$, the maximum weight perfect matching. Then, it is not hard to see that $w(M^*)$ is at least $\frac{1}{2}w(OPT)$ since the sets consisting of only the odd or even edges from $OPT$ are perfect matchings. Now, as a first step towards showing our black-box result, we provide a very general procedure in the following lemma that takes as input any arbitrary matching $M$, and outputs a Hamiltonian path (tour minus one edge) $Q$ whose weight is at least $\frac{3}{2}w(M)$ minus a small factor that vanishes as $N$ increases.

\begin{lemma}
\label{lem_tourcompletion}
Given any matching $M$ with $k$ edges, there exists an efficient ordinal algorithm that computes a Hamiltonian path $Q$ containing $M$ such that the weight of the Hamiltonian path in expectation is at least
$$[\frac{3}{2} - \frac{1}{k}]w(M).$$
\end{lemma}
\begin{proof}
We first provide the algorithm, followed by its analysis. Suppose that $K$ is the set of nodes contained in $M$.
\begin{enumerate}
\item Select a node $i \in K$ uniformly at random. Suppose that $e(i)$ is the edge in $M$ containing $i$.
\item Initialize $Q=M$.
\item Order the edges in $M$ arbitrarily into $(e_1, e_2, \ldots, e_k)$ with the constraint that $e_1=e(i)$.
\item For $j=2$ to $k$,
\item Let $x$ be a node in $e_{j-1}$ having degree one in $Q$ (if $j=2$ choose $x \neq i$) and $e_j=(y,z)$.
\item If $y >_x z$, add $(x,y)$ to $Q$, else add $(x,z)$ to $Q$.
\end{enumerate}

Suppose that $Q(j)$ consists of the set of nodes in $Q$ for a given value of $j$ (at the end of that iteration of the algorithm). We claim that $w(Q) \geq \frac{3}{2}w(M) - w(e_1)$, which we prove using the following inductive hypothesis,
$$w(Q(j)) \geq w(M) + \sum_{r=2}^j \frac{1}{2}w(e_r)$$

Consider the base case when $j=2$. Suppose that $e_1=(i,a)$, and $e_2=(x,y)$. Without loss of generality, suppose that $a$ prefers $x$ to $y$, then in that iteration, we add $(a,x)$ to $Q$. By the triangle inequality, we also know that $w(x,y) \leq w(x,a) + w(y,a) \leq 2w(x,a)$. Therefore, at the end of that iteration, we have
\begin{equation}
\label{eqn_temptspbb}
w(Q) = w(M) + w(x,a) \geq w(M) + \frac{1}{2}w(e_2).
\end{equation}

The inductive step follows similarly. For some value of $j$, let $x$ be the degree one node in $e_{j-1}$, and $e_j=(y,z)$. Suppose that $x$ prefers $y$ to $z$, then using the same argument as above, we know that $(x,y)$ is added to our desired set and that $w(x,y) \geq \frac{1}{2}w(y,z)$. The claim follows in an almost similar fashion to Equation~\ref{eqn_temptspbb} and the inductive hypothesis.

In conclusion, the total weight of the path is $\frac{3}{2}w(M) - w(e_1)$. Since the first node $i$ is chosen uniformly at random, every edge in $M$ has an equal probability ($p=\frac{1}{k}$) of being $e_1$. So, in expectation, the weight of the tour is $\frac{3}{2}w(M) - \frac{1}{k}w(M)$, which completes the lemma.
\end{proof}

The rest of the proof for the black-box mechanism follows almost directly from the lemma. Suppose that $M$ is an $\alpha$-approximation to the optimum matching $M^*$. Then, applying the lemma we get a path $Q$ whose weight is at least $[\frac{3}{2}-\frac{2}{N}]w(M)$. We complete this path to form a tour $T$ getting,
\begin{align*}
w(T) & \geq [\frac{3}{2}-\frac{2}{N}]w(M) &\geq [\frac{3}{2} - \frac{2}{N}]\frac{w(M^*)}{\alpha}\\
& \geq [\frac{3}{4} - \frac{1}{N}]\frac{w(OPT)}{\alpha}.
\end{align*}
Some basic algebra yields the desired bound.
\end{proof}

Using the above black-box theorem and the results of the previous sections, we immediately obtain the following: a $3.2$-approximation algorithms for Max $k$-Sum, a $4$-approximation algorithm for Densest $k$-Subgraph, and a $2.14$-approximation algorithm for Max TSP.

\section{Conclusion}
In this paper we study ordinal algorithms, i.e., algorithms which are aware only of preference orderings instead of the hidden weights or utilities which generate such orderings. Perhaps surprisingly,  our results indicate that for many problems including Matching, Densest Subgraph, and Traveling Salesman, ordinal algorithms perform almost as well as the best algorithms which know the underlying metric weights. This indicates that for settings involving agents where it is expensive, or impossible, to obtain the true numerical weights or utilities, one can use ordinal mechanisms without much loss in welfare.

While many of our algorithms are randomized, and the quality guarantees are ``in expectation", similar techniques can be used to obtain weaker bounds for deterministic algorithms (bounds of 2 for Matching, and of 4 for the other problems considered). 
It may also be possible to improve the deterministic approximation factor for matching to be better than 2: although this seems to be a difficult problem which would require novel techniques, such an algorithm would immediately provide new deterministic algorithms for the other problems using our black-box reductions. Finally, it would be very interesting to see how well ordinal algorithms perform if the weights obeyed some structure other than the metric inequality.
\bibliography{bibliography}
\bibliographystyle{plain}

\appendix
\newpage
\section{Friendship Networks}
The classic theory of Structural Balance~\cite{davis1977clustering} argues that agents embedded in a social network must exhibit the property that `a friend of a friend is a friend'. This phenomenon has also been observed in many real-life networks (see for example~\cite{goodreau2009birds}). Mathematically, if we have an unweighted social network $G=(V,E)$ of (say) friendships, it is easy to check if this property holds. If for any $(i,j) \in E$, $(j,k)$ is also an edge, then $(i,k)$ must also belong to $E$. For this reason, this property has also been referred to as transitive or triadic closure.

What about weighted networks that capture the `intensity of friendships'? There is no obvious way as to how this property can be extended to weighted graphs without placing heavy constraints on the weights. For example, if we impose that for every edge $(i,j)$, and every node $k$, $w(i,k) \geq \min(w(i,j), w(j,k))$, we immediately condemn all triangles to be isosceles (w.r.t the weights). Instead, we argue that a more reasonable mathematical property that extends triadic closure to Weighted graphs is the following

\begin{definition}{($\alpha$-Weighted Friendship Property)}
Given a social network $G=(V,E,W)$, and a fixed parameter $\alpha \in [0,\frac{1}{2}]$, for every $(i,j,k)$: $w(i,k) \geq \alpha [w(i,j) + w(j,k)]$.
\end{definition}

In a nutshell, this property captures the idea that if $(i,j)$ is a `good edge', and $(j,k)$ is a good edge, then so is $(i,k)$. The parameter $\alpha$ gives us some flexibility on how stringently we can impose the property. Notice that in a sense, this property appears to be the opposite of the metric inequality, here we require that $w(i,k)$ is not too small compared to $w(i,j) + w(j,k)$. However, we show that this is not the case; in fact for every $\alpha \geq \frac{1}{3}$, any weighted graph that satisfies the $\alpha$-Weighted Friendship Property must also satisfy the metric inequality.

\begin{claim}
Suppose that $G=(V,E,W)$ is a weighted complete graph that satisfies the $\alpha$-Weighted Friendship Property for some $\alpha \in [\frac{1}{3}, \frac{1}{2}]$. Then, for every $(i,j,k)$, we have $w(i,j) \leq w(i,k) + w(j,k)$.
\end{claim}
\begin{proof}
Without loss of generality, it suffices to show the proof for the case where $w(i,j)$ is the heaviest edge in the triangle $(i,j,k)$. Now consider $w(i,k)$,$w(j,k)$ and without loss of generality, suppose that $w(i,k) \geq w(j,k)$. Then, since the $\alpha$-friendship property is obeyed, we have
$$w(j,k) \geq \alpha[w(i,j) + w(i,k)] \geq \alpha w(i,j) + \alpha w(j,k).$$

This gives us that $w(i,j) \leq \frac{1-\alpha}{\alpha}w(j,k)$. It is easy to verify that for $\alpha \in [\frac{1}{3}, \frac{1}{2}]$, the quantity $\frac{1-\alpha}{\alpha} \leq 2$. Therefore, we get
$$w(i,j) \leq 2w(j,k) \leq w(i,k) + w(j,k).$$

\end{proof}

\section{Odd Number of Agents: Extensions}
In many of our algorithms, we assumed that $N$ (the number of agents) is even for the sake of convenience and in order to capture our main ideas concisely without worrying about the boundary cases. Here, we show that all of our algorithms can be extended to the case where $N$ is odd or not divisible by $3$ with only minor modifications to the proofs and bounds obtained.

\subsubsection{Matching} We begin by arguing that of all our algorithms and proofs for matching hold even when $N$ is odd. First of all, it is not hard to see that our framework does not really depend on the parity of $N$ and the lemmas on the greedy and random techniques carry over to the case when $N$ is not even. In particular, note that in Lemma~\ref{lem_randomlower}, we had that $E[w(M_R)] \geq \frac{|M_R|}{|E|}\sum_{x,y \in \mathcal{N}}w(x,y)$, where $E$ is the total number of edges in the complete graph. When $N$ is odd, $|M_R| = \frac{N-1}{2}$, and $|E|$ is still $\frac{N(N-1)}{2}$. Therefore, we still get that $E(w(M_R)] \geq \frac{1}{N}\sum_{x,y \in \mathcal{N}}w(x,y).$

Next, we argue that our main $1.6$-algorithm still holds for arbitrary $N$ (not divisible by two and/or three) with an $\epsilon$ multiplicative error term that vanishes as $N \to \infty$. In Algorithm~\ref{alg_bestmatching}, when $N$ is not divisible by three, we may need to choose a matching with $\ceil{\frac{N}{3}}$ edges for $M_0$. Let $B$ be the largest matching outside of $M_0$, then $|B|=\floor{\frac{N}{6}}$. Then, the second sub-routine in Algorithm~\ref{alg_bestmatching} proceeds by selecting $|B|$ edges from $M_0$, and matching those nodes arbitrarily to the nodes in $B$. Ideally, we would like $\ceil{\frac{N}{3}}=|M_0| = 2|B|=2\floor{\frac{N}{6}}$. The worst case multiplicative error happens when $\ceil{\frac{N}{3}}$ is much larger than $2\floor{\frac{N}{6}}$; this happens when $N$ is odd, and has the form $3p+2$ for some positive integer $p$. With some basic algebra, we can show that the multiplicative error is at most $\frac{7}{8(2N-3)}$, which approaches zero as $N$ increases.

\subsubsection{Max $k$-Sum}
In the case of the black-box theorem, the $2\alpha$ reduction still holds if we modify the algorithm as follows when $\frac{N}{k}$ is odd. Instead of selecting the optimum perfect matching, we need to select a matching of size $\frac{1}{2}(\frac{N}{k}-1)*k = \frac{1}{2}(N - k)$,  assign every pair of matched nodes to the same cluster, and the unmatched nodes arbitrarily. The rest of the proof is the same. Now when we apply this black-box result, we can no longer invoke the $1.6$-approximation algorithm for perfect matchings and only use the $2$-approximation greedy algorithm for a matching that selects $\frac{N-k}{2}$ edges, and therefore, the black-box result only yields a $4$-approximation when $\frac{N}{k}$ is odd, but our main algorithm gives a $2$-approximation algorithm irrespective of its parity.

\subsubsection*{Densest $k$-Subgraph}
All of the proofs for the Densest $k$-subgraph hold. When $k$ is odd, we simply select a matching with $\floor{\frac{k}{2}}$ edges instead of $\frac{k}{2}$. The rest of the proof is exactly the same.

\subsubsection*{Max TSP}
During the proofs for Max TSP, we extensively make use of the fact that $w(M^*) \geq \frac{w(T^*)}{2}$, where $M^*$ is the optimum matching, and $T^*$ is the optimal tour. This may no longer be true when $N$ is odd. However, suppose that $M^*_f$ is the optimum fractional matching, it is not hard to verify that $w(M^*_f) \geq \frac{w(T^*)}{2}$; this follows from taking $T^*$ and choosing each edge with probability $\frac{1}{2}$. Moreover, observe that all of our proofs in this paper for the greedy algorithm (namely Lemma~\ref{lem_matchinggreedy}) are true when we compare the solution to the optimum fractional matching of a given size. Therefore, the black-box result itself carries over.

\end{document}